\documentclass{collective-intelligence}
\pdfoutput=1

\newcommand*{\TITLE}{Macroscopes: models for collective decision making}
\newcommand*{\PDFAUTHOR}{%
Subramanian Ramamoorthy, Andras Z. Salamon, Rahul Santhanam}
\newcommand*{\KEYWORDS}{macroscope, collective decision making, communication complexity, meta-information}

\usepackage[final,pdfpagelabels=false,pdfpagelayout=SinglePage,bookmarksnumbered,hyperindex,colorlinks,linkcolor=blue,citecolor=blue,urlcolor=blue]{hyperref}
\hypersetup{pdftitle={\TITLE},pdfauthor={\PDFAUTHOR},pdfsubject={\TITLE},pdfkeywords={\KEYWORDS}}

\usepackage{tikz}
\usetikzlibrary{decorations.pathreplacing}

\newtheorem{theorem}{Theorem}

\title{MACROSCOPES: MODELS FOR COLLECTIVE DECISION MAKING}

\DeclareFixedFont{\auacc}{OT1}{ptm}{m}{n}{10}

\numberofauthors{3}
\author{
\alignauthor
Subramanian Ramamoorthy\\
       \affaddr{School of Informatics}\\
       \affaddr{University of Edinburgh}\\
       \affaddr{10 Crichton Street, Edinburgh}
       \affaddr{EH8~9AB, United Kingdom}\\
       \email{s.ramamoorthy@ed.ac.uk}
\alignauthor
{Andr{\auacc\'a}s Z. Salamon}\\
       \affaddr{School of Informatics}\\
       \affaddr{University of Edinburgh}\\
       \affaddr{10 Crichton Street, Edinburgh}
       \affaddr{EH8~9AB, United Kingdom}\\
       \email{andras.salamon@ed.ac.uk}
\alignauthor
Rahul Santhanam\\
       \affaddr{School of Informatics}\\
       \affaddr{University of Edinburgh}\\
       \affaddr{10 Crichton Street, Edinburgh}
       \affaddr{EH8~9AB, United Kingdom}\\
       \email{rsanthan@inf.ed.ac.uk}
}

\newcommand{\eps}{\epsilon}

\begin{document}

\maketitle

\begin{abstract}
We introduce a new model of collective decision making, when a
global decision needs to be made but the parties only possess partial
information, and are unwilling (or unable) to first create a global
composite of their local views.
Our macroscope model captures two key features of many real-world
problems: allotment structure (how access to local information is
apportioned between parties, including overlaps between the parties)
and the possible presence of meta-information (what each party knows
about the allotment structure of the overall problem).
Using the framework of communication complexity, we formalize the
efficient solution of a macroscope.
We present general results about the macroscope model, and also
results that abstract the essential computational operations
underpinning practical applications, including in financial markets
and decentralized sensor networks.
We illustrate the computational problem inherent in real-world
collective decision making processes using results for specific
functions, involving detecting a change in state (constant and step
functions), and computing statistical properties (the mean).
\end{abstract}

\section{Introduction}

We consider collective decision making processes such as a market
that acts as a central mechanism for coordinating the actions of
autonomous participants. We address the questions: how does one
measure the quality of the collective decision making process, and
how weak can the central market mechanism be? In many applications,
there is significant interest in decentralizing computation while
still being able to arrive at results that cannot be computed entirely
locally. We use a simple model to capture the informational complexity
of computing global functions by aggregating results from
participants\footnote{We use the common English words agent, party, participant,
and player synonymously, ignoring more specific usage.} who
are endowed with arbitrary allotments of local information. This
allows us to draw conclusions about the requirements on allotments
and protocols, for efficient collective information processing. A key
aspect of our model is the specification of meta-information based
on distinguishing perfect information, single-blind arrangements,
and double-blind arrangements. Our technical framework is built on
the notions of communication complexity. We assume that participants
possess information which is not available to other participants;
we call this the \emph{private} information.

This work is motivated by several applications.

A rather timely application is found in the domain of
participants in electronic markets. Often, such as in financial
markets, participating agents would benefit from an understanding of the
global system dynamics \cite{darley2007nasdaq}.\footnote{See
\url{http://www.bankofengland.co.uk/publications/speeches/2009/speech386.pdf}
for a discussion of this issue from the perspective of
financial regulation, and
\url{http://www.bis.gov.uk/foresight/our-work/projects/current-projects/computer-trading}
for information on a major study by the UK government, under the
Foresight Project.}
For instance, agents might like to have signals
that indicate the presence of herding, bubbles and other aggregate
phenomena. Typically, the local view of a single agent does not provide
sufficient information to reliably detect this. Moreover, in such a
domain, one is tightly constrained by what information can be revealed,
incentives to reveal this information, and other aspects related to
privacy in computation. If we seek efficient decentralized information
processing mechanisms under these constraints, then we would like to be
able to determine what is or is not possible, employing only a coarse characterization
of resources and endowment of information. Recent studies such as
in anonymized financial chat rooms \cite{Lu2011:strategic} provide
interesting insights into the behaviour of such collectives, such
as the characterization of equilibria in which a subset of traders
profit from the information of others. This is but one example of a
larger body of economic literature related to phenomena in networked
markets~\cite{Hurwicz2006:designing}. However, in that literature,
it is not typical to investigate our question of how the allotment
structure and communication protocols relate to the efficiency with
which specific types of computation are achieved. For instance,
change detection \cite{Basseville93} is of fundamental importance
in financial markets -- how weak a protocol is sufficient to decide
a change has occurred? Recent work on the topic of complexity of financial
computations by \citeasnoun{Arora2011:computational} indicates that this is a
fertile direction to pursue.

Similar issues arise in many other application domains, such as
mobile sensor networks and distributed robotics. \citeasnoun{Leonard07}
describe a mobile sensor network for optimal data gathering,
using a combination of underwater and surface level sensing robots
to optimally gather information such as chemical and thermal
dynamics in a large volume of water (typically measured in square miles).
Similar systems have been utilized for tracking
oil spills and algal blooms. A key computational method utilized by
such distributed robotic networks involves distributed optimization
\cite{Bullo09}. The deployment of modules in such a network needs to
satisfy a spatial coverage requirement, described as a
cost function, so that each module plans trajectories to optimize this
criterion. The sensor fusion problem, to determine a combined estimate
of an uncertain parameter, may also be posed an as optimization
problem in the sense of maximizing information gain. Despite this
rigorous approach, relatively little is known about how to compare
different formulations of these optimization problems -- given that
we are interested in a certain type of global function (say, number
of peaks in a chemical concentration profile or some distributional
aspect of the overall field) using weak local sensing and the ability
to move sensing nodes, how does one compare or otherwise characterize
protocols and other aspects of the problem formulation?

A line of work that begins to touch upon some of these questions is
that of Ghrist and collaborators. \citeasnoun{Ghrist09} use tools
from algebraic topology to solve the problem of counting targets
using very weak local information (such as unlabelled counts in a
local neighbourhood). \citeasnoun{deSilva07} present an approach
to detection of holes in coverage through decentralized computation
of homology. Here again, the focus being on aspects of the specific
function being computed, the authors do not address the relationship
between the protocols and problem formulation, and the efficiency
of computation.

Extending the idea of decentralized computation in social systems,
consider the problem faced by a program committee, such as one
that might review this paper. We seek a decentralized computation
of a ranking problem. Similar ranking problems also occur in
executive decision making such as the hiring decision in academic
departments. The key issue here is that of parsimonious information
sharing, coupled implicitly or explicitly with the meta-information
problem.  These challenges arise due to limitations on the
capacities of the decision makers to exchange information with each other.

The common theme underlying all of these applications is the
computation of a function based on allotment of portions of the
information to parties who have reasonable amounts of computational
resource but would like to keep exchange of information limited. We
wish to understand how weak the corresponding protocols can be,
for various types of functions. Major categories of functions of
interest include change detection and ranking.  We model change
detection by an abstract version of the key underlying problem, of
determining whether the data forms a constant or a step function.
The main statistical property we consider here is the computation
of the mean.

We are interested in understanding just how much communication must
occur to answer various questions of interest.  Therefore, instead of
working with detailed models for the questions about market behaviour
or sensor networks discussed previously, we have deliberately kept
the models we study as \emph{simple} as possible.  \emph{This makes
our lower bound results stronger.}  Determining an answer to any
more realistic question will require even more information to be
exchanged than in these simple models, as long as the more realistic
model includes the simpler problem at its core.  It therefore makes
sense in our setting to study the simplest possible embodiment of
each of the core problems.  For the upper bounds, our results are a
first step and will need to be extended to more realistic models.

\section{Model}

Our model is based on the notion of \emph{communication complexity}
\cite{Kushilevitz-Nisan97}, which has been highly influential in
computer science.  A \emph{Boolean} function models yes/no decisions,
by requiring that the function take either the value 0 or the value 1.
A quantity with value either 0 or 1 is known as a \emph{bit}, and
quantities that are drawn from a larger range of values can be expressed by
using multiple bits; a function that is defined over a domain containing $2^n$ different
values is said to have an $n$-bit \emph{input}.
Say two players Alice and Bob wish to compute a Boolean function $f$
on a $2n$-bit input, but Alice only has access to the first $n$ bits and Bob
to the other $n$ bits. Alice and Bob are not \emph{computationally}
constrained, but they are \emph{informationally} constrained. The
question now is: how many bits of information do Alice and Bob need to
exchange to compute $f$ on a given $2n$-bit input? A \emph{protocol}
for this problem specifies, given the inputs to the players and the
communication so far, which player is the next to send information,
as well as what information is actually sent. There is a trivial
protocol where Alice merely sends her part of the input to Bob. Bob
now has all the information he needs to compute $f$, and he sends back
the 1-bit answer to Alice. The \emph{cost} of a protocol is the
total number of bits that are exchanged; this simple protocol has a
cost of $n+1$ for \emph{any}
function. The field of two-party communication complexity studies, for
various functions $f$ of interest, whether more efficient protocols
exist. As an example, for the Equality function which tests whether
Alice and Bob's inputs are exactly the same, it is known that the
communication upper bound of $n+1$ is tight for \emph{deterministic}
protocols, but there is an improved protocol with cost $O(\log(n))$
when the players' messages are allowed to be randomized and it is
sufficient for the final answer to be correct with high probability.

The notion of communication complexity can be generalized
from the two-party setting to the \emph{multi-party} setting
\cite{Chandra1983:multi-party}. Here the number of
players is not limited to two, each player has some information about the global input, and
they wish to compute some Boolean function of the global input. There
are two standard models for how the input is distributed among the
players: the number-in-hand model (NIH) and the number-on-forehead
(NOF) model. Suppose there are $k$ players and the global input is $N$
bits long. In the NIH model, there is some fixed partition of the
global input into $k$ parts, and each player gets one of these parts.
In the NOF model, again there is a fixed partition into $k$ parts,
but the $i$-th player gets all the parts \emph{except} the $i$-th part.

The main motivation for our model is that in many situations, such
as financial markets or sensor networks, information is distributed
among players in a more complex fashion than in the NOH or NIF
models. Moreover, the players might not have control over which pieces
of information they have access to -- the \emph{allotment} of inputs to
players might be arbitrary, perhaps even adversarial.
As an example, creators of financial instruments may decide
which assets to bundle into pools that are then offered for sale.
Purchasers of such instruments might wish to check fairness of
allocation, without revealing to each other their precise holdings;
or regulators might wish to check that sellers behaved impartially
but without relying on full disclosure.

Yet, the players might
still wish to compute some function of their global input
in this less structured setting. Now different kinds of question arise
than in the standard communication complexity setting. For a given
function, which kinds of allotment structures allow for protocols? Does
the meta-knowledge of what the allotment structure actually is make
a difference to whether there is an efficient protocol or not? These
questions are interesting even for simple functions which have been
thoroughly investigated in the standard setting.

To be more formal, let $f$ be a function which $k$ players wish
to compute on a global input $x$ of size $N$ bits. An $(N,k)$
allotment structure is a sequence of $k$ subsets $S_1, S_2 \ldots
S_k$ of $[N] = \{1,2,\dots,N\}$. An \emph{allotment structure} corresponds to an allotment of
input bits among players in the following way: Player $i$ receives
all bits $x_j$ for $j \in S_i$. Note that unlike in the NOH and
NIF models, this allotment of input bits is completely general --
it might be the case that two players receive the same set of bits,
for example. This is the main novelty of our approach. Our intention
here is to model two kinds of situations. In the first,
the players have little control over which pieces
of information they can access -- they have to do the best they can,
with the available information. In the second, the allotment is made
by a centralized authority, and it is of interest to study \emph{which}
allotment would most facilitate the computation in question.

A $k$-player \emph{macroscope} on $N$ bits is simply a function $f$ on $N$ input bits together with
an $(N,k)$ allotment structure.  We will abuse notation and sometimes
use a macroscope to refer to a \emph{sequence} of functions $f_{N},
N = 1 \ldots \infty$, where each function $f_N$ depends on $N$
bits. This will enable us to pose and study the question of the
asymptotic efficiency of protocols for macroscopes.

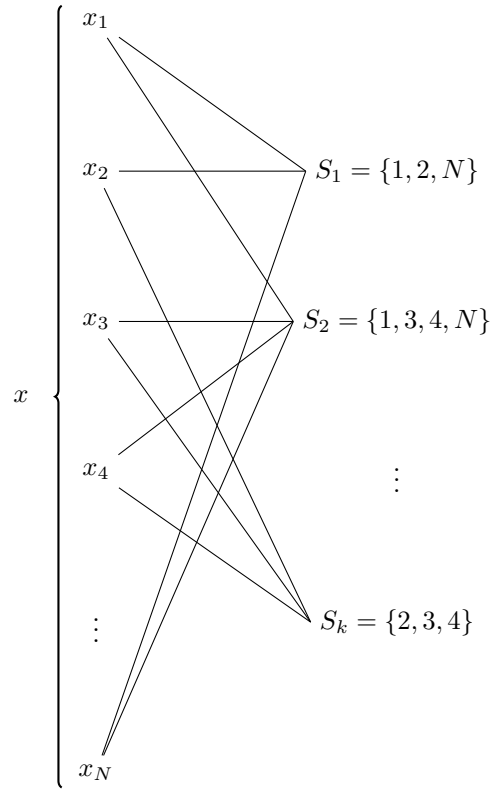
\begin{figure}[h]
\centering
\begin{tikzpicture}[xscale=2,yscale=-2]
\node (x1) at (0,0) {$x_1$};
\node (x2) at (0,1) {$x_2$};
\node (x3) at (0,2) {$x_3$};
\node (x4) at (0,3) {$x_4$};
\node (xdots) at (0,4) {$\vdots$};
\node (xN) at (0,5) {$x_N$};
\node (S1) at (2,1) {$S_1 = \{1,2,N\}$};
\node (S2) at (2,2) {$S_2 = \{1,3,4,N\}$};
\node (Sdots) at (2,3) {$\vdots$};
\node (Sk) at (2,4) {$S_k = \{2,3,4\}$};
\draw (x1) -- (S1.west);
\draw (x2) -- (S1.west);
\draw (xN) -- (S1.west);
\draw (x1) -- (S2.west);
\draw (x3) -- (S2.west);
\draw (x4) -- (S2.west);
\draw (xN) -- (S2.west);
\draw (x2) -- (Sk.west);
\draw (x3) -- (Sk.west);
\draw (x4) -- (Sk.west);
\node (x) at (-0.5,2.5) {$x$};
\draw[thick,decorate,decoration={brace,mirror,raise=1.3em}] (0,-0.1) -- (0,5.1);
\end{tikzpicture}
\caption{%
Example $(N,k)$ allotment structure.  We wish to compute $f(x)$
with a $k$-player protocol, in two situations.
(Single-blind) Player $i$ knows $\{S_1,\dots,S_N\}$ and $x_j$ for $j \in S_i$.
(Double-blind) Player $i$ only knows $x_j$ for $j \in S_i$.
}
\label{fig:allotment}
\end{figure}

The generalized modelling of the allotment of inputs raises the
issue of \emph{meta-information} -- how much do players know about the
allotment of inputs, and how can they take advantage of this? In the
case of the NIH and NOF models, the allotment is implicitly known to
all players because it is fixed in advance. However, in our setting,
there are two different kinds of situations -- the \emph{single-blind}
situation and the \emph{double-blind} situation. In a single-blind
macroscope, all players know the allotment structure, however
Player $i$ does not know the values of any input bits apart from
the ones whose indices are in $S_i$. In a double-blind macroscope,
the players are more hampered in that they do not even know the
allotment structure, however they do know the indices of the bits they
receive.

It remains to formally define what a \emph{protocol} is in our model. To
keep things simple, we focus on \emph{simultaneous-message} protocols,
where each player broadcasts a sequence of bits to all players; this is
often presented figuratively as each player writing their bit string on a universally viewable
blackboard. A protocol \emph{solves} a macroscope if each player can
determine the value of the function on the global input simply by
looking at its own input bits as well as the information written on
the blackboard. The cost of a protocol for a macroscope is then the total
length of strings written by the players. In protocols for single-blind
macroscopes, the message of Player $i$ is a function of the values of
bits whose indices are in $S_i$ as well as of the allotment structure.
For double-blind macroscopes, the message of Player $i$ is a function
only of the values of bits whose indices are in $S_i$.

We make a deliberate choice in our modelling to be highly general in
terms of the allotment structure, and to be specific in terms of the
structure of the actual communication. This is because our main aim
is to understand the impact of the allotment structure on efficiency
of communication.  Our model can be extended to allow more degrees of
freedom with regard to the communication structure. One way in which
this can be done is to allow multiple-round protocols, where players
communicate in turns, with the protocol specifying whose turn it is to
communicate. Another is by allowing randomness -- here each player is
assumed to have access to a private source of randomness, on which
its message can depend. A third way is to restrict communication
to take place between specified pairs of players, i.e., there is an
implicit topology of communication. This third approach is taken in
the field of distributed algorithms \cite{Lynch96}, where however
input allotment is not modelled in a flexible way.

We are interested in protocols which have communication as low as
possible.  This is desirable not just in terms of efficiency, such as
meeting bandwidth constraints, but also in terms of \emph{privacy}. In
applications such as financial markets, the players would like to
obtain some global knowledge without revealing their own inputs. Thus,
Player $i$ has more than one reason for not following the trivial
protocol of publishing the values of all bits in $S_i$. The \emph{lower}
the communication, the \emph{less} the information revealed about the
values of bits held by individual players; we will rely on this link
between parsimony of communication and the weakness of the
coordination mechanism. Privacy requirements are
modelled more explicitly in sub-areas of cryptography such as secure
multi-party computation \cite{Yao82,Goldreich04}. We prefer not
to model these requirements explicitly so as not to complicate our model
too much.

We make no assumption about the relationship between the number of
players and the number of bits in the global input. In an application
such as sensor networks, there might be few players (sensors), each
having a large amount of information, whereas in the financial markets
application, there are typically many players each having few pieces
of information. Our model deals equally well with both extremes.

A first observation is that to compute a non-trivial function over
the global input, i.e. a function that depends on all the input bits,
the allotment structure must satisfy the \emph{covering} property --
each index $j \in [N]$ lies in at least one set $S_i$ of the allotment
structure.  If the covering property did not hold, consider an index
$j$ which does not belong to the allotment, and an input $X$ such
that $f$ is \emph{sensitive} to $X$ at index $j$, meaning that $f(X)$
is different from $F(X^{flip}_{j})$, where $X^{flip}_{j}$ is $X$ with
the value of the $j$th bit flipped. By the non-triviality of $f$,
such an input $X$ must exist. Clearly any protocol outputs the same
answer for $X$ as for $X^{flip}_{j}$ since $j$ does not belong to
the allotment, and hence the protocol cannot be correct. Henceforth, we 
automatically assume that a macroscope has the covering property.

There are no general necessary conditions on the allotment structure
beyond the covering property for computation of non-trivial functions.
But intuitively, the more ``even'' the allotment is, in the sense of
each bit being allotted to the same number of players, the easier it
is to compute a symmetric function of the inputs. We define an \emph{even} $(N,k)$ allotment structure as an allotment structure for which
there is a number $C$ such that each index $i \in [N]$ belongs to
exactly $C$ distinct subsets $S_i$, and each subset $S_i$ is of the
same size. Clearly, for such an allotment structure, each set $S_i$
is of size $NC/k$.

\section{Results}

Our first results address the question of what we can say in general about
the cost of single-blind and double-blind macroscopes.

\begin{theorem}
\label{gen-single-blind}
Every single-blind macroscope on $N$ bits has a protocol with cost $N$. Moreover, this bound is optimal.
\end{theorem}

Note that the upper bound \emph{does not} depend on the number of players.
The proof of the upper bound in Theorem~\ref{gen-single-blind} takes advantage of the global
knowledge the players have about the allotment structure. 

\smallskip
\begin{proof}
Consider a protocol
in which each input bit $X_i$ has a player ``responsible'' for it -- Player
$j$ is responsible for input bit $X_i$ iff $i \in S_{j}$ and $i \not \in S_{k}$
for $k < j$. It follows from the covering property that each input bit
has a player responsible for it. It should also be clear that at most one 
player is responsible for any given input bit. The protocol consists of
players sending the values of all the bits they are responsible for. Since the 
macroscope is single-blind, each player knows who is responsible for which
input bits. Hence each player can reconstruct the input from the information
 sent in the protocol, and therefore also compute the function on the input.
The cost of the protocol is $N$, since each bit has exactly one player
responsible for it. 

To see that this bound is optimal, consider the macroscope consisting of the
Parity function\footnote{The Parity function returns 1 if an odd number
of the input bits are 1, and 0 otherwise.} together with an $(N,N)$ allotment structure which allots
each input bit to a distinct player. Suppose there is a protocol for this
macroscope where one of the players does not send a message. Assume,
without loss of generality,
that this player holds the $i$-th bit. Then the Parity function cannot
depend on the $i$-th bit, which is a contradiction.
\end{proof}

\begin{theorem} \label{gen-double-blind}
Every $k$-player double-blind macroscope on $N$ bits 
has a protocol with cost at most $2Nk$. 
\end{theorem} 

\smallskip
\begin{proof}
The protocol giving the upper bound of Theorem~\ref{gen-double-blind} is a simple one. The $j$-th
player sends an $N$-bit string specifying its allotment $S_j$, as well as
at most $N$ bits which specify the values of the bits whose indices are in
the allotment $S_j$. Once this information is made public by each player,
the players can each reconstruct the input and hence compute the function
on that input.
\end{proof}

We do not know whether the bound in Theorem \ref{gen-double-blind} is
optimal in general, but we believe this to be the case for protocols
that use only one round of communication.

Our primary focus is on studying the complexity of macroscopes for
problems which arise in contexts such as electronic markets and
distributed sensor networks. For each of these problems, we are interested
in issues such as the optimal cost of solving a macroscope, the differential
cost of meta-information (the reduction in cost when using a
single-blind macroscope rather than a double-blind macroscope), and for
single-blind macroscopes, the dependency of the cost on the
allotment structure.

A fundamental problem in the context of electronic markets is the
Constancy problem of detecting whether a given function is constant or
not.\footnote{Although the actual values may be 
constantly changing, we use the abstract Boolean version of this problem to represent a more 
significant shift from some statistical notion of the normal
behaviour.  For instance, we might ask if the observed deviations from some base 
value are significant at a 5\% confidence level.}
We model
this problem as a Boolean function on $[D]^N$, which is 1 if all the
inputs are equal, and 0 otherwise. This requires a slight adaptation of our model to inputs 
which are $D$-ary rather than binary, but this adaptation can be done
in a natural way. We are able to characterize the cost
of protocols for single-blind Constancy macroscopes \emph{optimally} in terms of the allotment structure.

Given an $(N,k)$ allotment structure, we define the intersection graph of the
structure as follows. The graph has $N$ vertices, and there is an edge
between vertex $i$ and vertex $j$ for $i,j \in [N]$ iff $S_i \cap S_j \neq \emptyset$. 

\begin{theorem}
\label{Single-Blind-Constancy}
Every $k$-player single-blind Constancy macroscope on $N$ $D$-ary inputs can be 
solved with cost 
$r (\lceil \log(D) \rceil) + k$, where
$r$ is the number of connected components of the intersection graph of
the allotment structure associated with the macroscope. Moreover, this
bound is optimal to a constant factor.
\end{theorem}

Note that the bound does not actually depend on $N$, the number of inputs!

\medskip
\begin{proof}
Since the macroscope is single-blind, each player knows the identities of all
the other players in the same connected component in the intersection graph
of the allotment structure. This is because these identities depend only
on the allotment structure and not on the input.

The protocol is as follows: For each connected component, it is only the
player with the smallest index in that connected component who sends a ``long''
message of $\lceil \log(D) \rceil$ bits long, specifying a value in $[D]$ that
occurs in its portion of the input. In addition, each player sends a 1-bit
message saying whether its portion of the input is constant or not. The players
know that Constancy holds if each 1-bit message encodes ``yes'', and in 
addition, the values in $[D]$ sent in the long messages are all the same.
Indeed, if the function is constant, it is clear that all the 1-bit messages
are ``yes'', and that the values in the long message are the same. To argue
the converse, just notice that if a fixed player in a 
connected component of the intersection graph sees a constant value $l \in [D]$, and if all other players in the connected component see a constant value,
it follows that all players in the connected component see the same
constant value $l$.

We argue that this bound is optimal up to a factor of 2. We will give separate
arguments that $k$ bits of communication are required and that $r \lceil \log(D) \rceil$ bits of communication are required. From these separate bounds, it
follows that $max(k, r \lceil \log(D) \rceil) \geq (k+r \lceil \log(D) \rceil)/2$ bits of communication are required.   

To see that $k$ bits of communication are necessary, consider any allotment
structure such that  $S_i \setminus (\cup_{j \neq i}
S_j)$ is non-empty for each player $i$. Define a function $f: [k]
\rightarrow [N]$ such that $f(i) \in S_i \setminus \cup_{j \neq i} S_j$
for each player $i$. Suppose there is a player $i$
who does not send a message. Then the protocol gives the same answer
on both the all 1's input and the input that is all 1 except at $f(i)$, since
the communication pattern is the same for both of these inputs. However the
Constancy function differs on these inputs.

Next we show the $r \lceil \log(D) \rceil$ lower bound. Consider any
allotment structure whose intersection graph contains $r$ connected components.
Suppose that fewer than $r \lceil \log(D) \rceil$ bits of communication are 
sufficient for solving the macroscope on this allotment structure. Then there is some connected component $C$ of the intersection graph
such that the corresponding players send less than $\log(D)$ bits
of communication in all. This implies that there are two values $v_1, v_2 \in [D]$ such that the communication pattern of players in $C$ is exactly the same
when the players in $C$ all receive the input $v_1$ as when they all receive
the input $v_2$. Now consider two inputs -- input $x$ in which all co-ordinates
are the constant $v_1$ and input $y$ in which all co-ordinates outside
$C$ have value $v_1$ and co-ordinates in $C$ have value $v_2$. The communication pattern of the protocol is the same for $x$ and $y$, however the Constancy
function is true for $x$ and false for $y$. This is a contradiction. 
\end{proof}

Thus, for single-blind Constancy macroscopes, the critical property of the allotment
structure is the number of connected components of the intersection graph.
The fewer the number of connected components, the more efficiently the
macroscope can be solved. We next study the situation for double-blind
macroscopes.

\begin{theorem}
\label{Double-Blind-Constancy}
Every $k$-player double-blind Constancy macroscope on $N$ $D$-ary inputs can be
solved with cost $k \lceil \log(D+1) \rceil$. Moreover, there are $k$-player
double-blind Constancy macroscopes which require cost $k \lceil \log(D) \rceil$.
\end{theorem}

\begin{proof}
The protocol giving the upper bound is simple. Each player sends a message
encoding one of $D+1$ possibilities: either the players' portion of the input
is non-constant, or if it is constant, which of the $D$ possible values it is.
The protocol accepts if each message encodes the same value $v \in [D]$.

For the lower bound, since each player is unaware of players' allotments
other than its own, the lower bound of $ r \lceil \log(D) \rceil$ in the proof
of Theorem \ref{Single-Blind-Constancy} holds with the maximum possible
value of $r$, namely $r = k$.
\end{proof}

Thus, in the case of the Constancy function, the differential cost of
meta-information can be quite significant, especially when the intersection
graph of the allotment structure is connected.

Next we consider a formalization of the
change detection problem. The Boolean
Step Function (BSF) problem is defined as follows: a string
$x \in \{0,1\}^{N}$ evaluates to 1 if there is an index $i$ such that
$x_j = 0$ for $j \leq i$ and $x_j = 1$ for $j > i$, or to 0 otherwise.

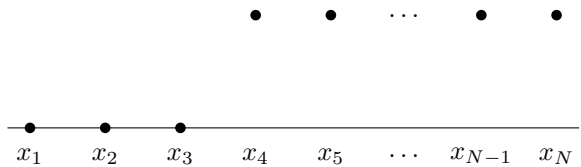
\begin{figure}[h]
\begin{tikzpicture}[xscale=1,yscale=1.5]
\tikzstyle{v}=[shape=circle,fill,inner sep=0pt,minimum size=4pt]
\draw[thin] (0.7,1) -- (8.3,1);
\node at (1,0.75) {$x_1$};
\node at (2,0.75) {$x_2$};
\node at (3,0.75) {$x_3$};
\node at (4,0.75) {$x_4$};
\node at (5,0.75) {$x_5$};
\node at (6,0.75) {$\dots$};
\node at (7,0.75) {$x_{N-1}$};
\node at (8,0.75) {$x_N$};
\node[v] at (1,1) {};
\node[v] at (2,1) {};
\node[v] at (3,1) {};
\node[v] at (4,2) {};
\node[v] at (5,2) {};
\node at (6,2) {$\dots$};
\node[v] at (7,2) {};
\node[v] at (8,2) {};
\end{tikzpicture}
\caption{Boolean step function (BSF).}
\label{fig:step}
\end{figure}

Possibilities for different structures of allotment lead to a twofold
challenge in a BSF macroscope.  First, no party may see the step,
so parties need to share some information about the values they see.
Second, if several parties detect a step, they then need to determine
whether they are observing the same step.

\begin{figure}[h]
\begin{tikzpicture}[xscale=1,yscale=1.5]
\tikzstyle{v}=[shape=circle,fill,inner sep=0pt,minimum size=4pt]
\draw[thin] (0.7,1) -- (9.3,1);
\node at (1,0.75) {$x_1$};
\node at (2,0.75) {$x_2$};
\node at (3,0.75) {$x_3$};
\node at (4,0.75) {$x_4$};
\node at (5,0.75) {$x_5$};
\node at (6,0.75) {$\dots$};
\node at (7,0.75) {$x_{N-2}$};
\node at (8,0.75) {$x_{N-1}$};
\node at (9,0.75) {$x_N$};
\node[v] at (1,1) {};
\node[v] at (2,1) {};
\node[v] at (3,1) {};
\node[v] at (4,2) {};
\node[v] at (5,2) {};
\node at (6,2) {$\dots$};
\node[v] at (7,2) {};
\node[v] at (8,1) {};
\node[v] at (9,2) {};
\end{tikzpicture}
\caption{A Boolean function that is not a step function.}
\label{fig:nonstep}
\end{figure}
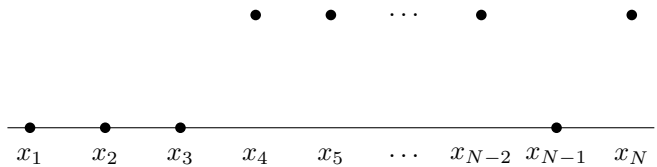

\begin{theorem}
\label{double-blind-BSF}
Every $k$-player double-blind BSF macroscope on $N$ bits can solved with cost
$2k \lceil \log(N) \rceil$.
\end{theorem}

\begin{proof}
The protocol is as follows: for each $i \in [k]$, Player $i$ sends two indices $l(i)$ and $m(i)$, 
where
$l$ is the largest index in $S_i$ for which $x_l = 0$ and $m$ is the smallest
index in $S_i$ for which $x_m = 1$. Given all these messages, each player
can calculate the value of the smallest index $m$ for which $x_m = 1$ simply
by taking the minimum of $m(i)$ over all players $i$, as well as the largest
index $l$ for which $x_l = 0$, simply by taking the maximum of $l(i)$ over
all players $i$. Note that $BSF(x) = 0$ iff $l=m-1$.
\end{proof} 

We conjecture that the bound of Theorem \ref{double-blind-BSF} is 
tight for double-blind BSF macroscopes for protocols that use only one
round of communication. For single-blind
macroscopes, however, we can do better.

\begin{theorem}
\label{single-blind-BSF}
Every $k$-player single-blind BSF macroscope on $N$ bits can be solved with
cost $k \lceil \log(N) \rceil + 2k$.
\end{theorem}

\begin{proof}
The protocol witnessing the upper bound is as follows. Each player sends
a message consisting of two parts. The first part is 2 bits long, and 
specifies which of the following is the case: (1) the player's portion of the 
input is constant, (2) there is a single transition from 0 to 1 in the player's input, (3) neither (1) nor (2) holds. The second part is $\lceil \log(N) \rceil$ bits long. The interpretation of the second part of the message is as follows:
if case (1) holds for the first part, then the second part encodes which
constant (either 0 or 1) the player is given. If case (2) holds, then the
second part encodes the index at which a transition occurs, i.e., a number
$j \in [N]$ such that $j \in S_i$ (assuming that the player in question is
Player $i$) and such that for all $l \in S_i$, $l \leq j \implies x_j = 0$ and
$l > j \implies x_j = 1$. If case (3) holds, the contents of the second
part of the message are irrelevant.

From the messages, the players can either reconstruct the input $x$, if 
case (1) or (2) holds for each player, or conclude directly that $BSF(x) = 0$, 
if case (3) holds for any player. From the input $x$, each player can
compute $BSF(x)$ on its own.
\end{proof}

In this case, too, the advantage of using single-blind macroscopes can be
seen, though a matching lower bound in Theorem \ref{double-blind-BSF} is needed
to prove this.

Next, we attempt to model the \emph{averaging} function. Distributional
statistics of different kinds need to be computed in a decentralized way in
various contexts such as sensor networks. To model this, we again depart
from the framework of bit strings as inputs. The input is now a sequence of
real numbers $\{x_i\}$, $x_i \in [0,1]$ for each $1 \leq i \leq N$. Given
a parameter $\epsilon > 0$, we study
the cost of protocols for $\eps$-Averaging macroscopes, where each player
needs to arrive at an $\epsilon$-additive approximation to the average of the
numbers $x_i$ by using the protocol to communicate.

\begin{theorem}
\label{single-blind-Averaging}
Let $\epsilon > 0$ be fixed. Every $k$-player single-blind $\eps$-Averaging 
macroscope on
$N$ inputs can be solved with cost $k \lceil \log(k/\epsilon) \rceil$. 
\end{theorem}

Notice again that there is no dependence of the cost on $N$, merely on the
number of players and the approximation error.

\begin{proof}
We again use critically the meta-information of players about the
allotment
structure. Each player $i$ knows, for each index $j \in S_i$, the number
$N_j$ of distinct players receiving input $x_j$. The message sent by Player
$i$ is an $\epsilon/k$-additive approximation to the quantity $\Sigma_{j \in S_i} x_j/N_j$. Since the quantity is between $0$ and $1$, the approximation can be specified using $\lceil \log(k/\epsilon) \rceil$ bits. Given the messages of
all players, each player can compute an $\epsilon$-approximation to the average
simply by summing all the individual approximation. Since the individual
approximations are $\epsilon/k$-additive approximations, the sum will be
an $\epsilon$-approximation to the average.
\end{proof}

In the case of Averaging macroscopes, we \emph{can} show that the double-blind
restriction is a significant one, in that it leads to a dependence of 
the cost on the number of players. The proof uses a dimensionality argument.

\begin{theorem}
\label{double-blind-Averaging}
Let $\epsilon > 0$ be a parameter. There are $2$-player double-blind
$\epsilon$-Averaging macroscopes on $N$ inputs which require cost $N \log(1/(N \epsilon))$. 
\end{theorem}

\smallskip

\begin{proof}
Consider a $2$-player double-blind macroscope for $\eps$-Averaging, with
Player 1 receiving an allotment $S_{1} \subseteq [N]$ and Player 2
receiving an allotment $S_{2} \subseteq [N]$. Our assumption that the allotment
structure is covering implies $S_{1} \cup S_{2} = [N]$, but neither player
knows anything about the other player's allotment beyond this fact.

We show that a player's message must essentially reveal its input. Consider
Player 1, and let $|S_1| = l$. Now for each $i \in S_1$, consider $S_{2,i} = [N] \setminus \{i\}$. If Player 1 has allotment $S_1$ and Player 2 has allotment
$S_{2,i}$, then in a correct protocol, Player 2 eventually knows an $\epsilon$-approximation to $\Sigma_{j \in [N]} x_j/N$. Since it also knows all values except
$x_i$, this implies that it knows an $N \epsilon$-approximation to $x_i$.
Thus an $N \epsilon$-approximation to $x_i$ should be extractable from 
Player 1's message. This should hold for every $i$, which means Player 1
must send at least $l \log(1/(N \epsilon))$ bits. By a symmetric argument,
Player 2 must send at least $(N-l) \log(1/(N \epsilon))$ bits, which means
that at least $N \log(1/(N \epsilon))$ bits are communicated in all. 
\end{proof}

Theorem \ref{double-blind-Averaging} has the disadvantage that it does not
say much about $\epsilon$, which may be large in comparison to $N$. However, we
believe that a refinement of the proof which argues about information revealed
about subsets of inputs rather than individual inputs can be used to
establish an improved lower bound.

\vspace*{-5pt}
\section{Discussion}

The previous two sections have focused on the details of a
model for collective information processing and properties of this
model. We now take a step back to discuss why these results are
of relevance to the motivating examples identified in the Introduction.

Consider the problem of agents in financial markets. It is increasingly
the case that, with the emergence of diverse communication methods
and agents deciding at time scales ranging from microseconds to days,
common knowledge isn't so commonly available in practice. This has
significant implications for dynamics \cite{Rubinstein89} and much
of the modern discussion regarding markets is related to such issues
\cite{Arora2011:computational}. In this setting, there is a need for
diagnostic tools that could provide useful signals -- by computing
global properties, i.e. functions, based on local information that
can be used subject to limitations on protocols. For instance, has
there been a change from a `constant' level in a global sense? Or,
is there a significant difference -- in the form of a step change
-- between segments of a networked market? We illustrate the use of
concepts from complexity theory to address abstract versions of such
questions. Indeed, our techniques could be used to answer further
such questions -- about statistical distributions, ranking queries,
etc. A key novelty, even in comparison to the state of the art in
communication complexity, is that we consider an arbitrary endowment
of inputs and meta-information such as single and double blind protocols.

An important general direction for future work in this area would be to extend
our analysis to more directly address the subtleties of the above
mentioned dynamics as they occur in applications of interest. Also,
we would like to better understand the relationships between our
model of decentralized computation under informational constraints
and previous established models, such as \cite{Hurwicz2006:designing},
which employ different methods and focus more on temporally extended
sequential protocols such as auctions.

In terms of more specific questions, our model could be extended in
various ways. We have considered the case of a \emph{simple} structure of communication,
with simultaneous messages sent in one round of communication, and a 
possibly complex allotment structure of the inputs. We could allow
more sophisticated communication structure. For example, double-blind
protocols with two rounds of communication can emulate single-round single-blind protocols with some loss in efficiency, simply by using the first round to
share publicly information about the input allotment, and then running
the single-blind protocol in the second round. More generally, communication
might be restricted to occur only between specific \emph{pairs} of players.
In the context of sensor networks, for instance, it is natural to model
both the location of information and the structure of communication as
governed by the topology of the ambient space. Even studying very simple
functions such as Constancy in such general models appears to be interesting.

We emphasize that we are interested in understanding the communication
requirements of even very simple functions in modelling frameworks that
render them non-trivial. This distinguishes our work from the existing
research on communication complexity, where functions such as Constancy and
BSF are trivial because the model of communication is so simple. We are
especially interested in modelling some aspects of collective information
processing, such as information overlap and meta-information, which have been 
neglected so far and
which we believe can have a significant impact on the efficiency of
communication. The way we capture these notions is quite flexible and can
be used both to model computation of continuous quantities such as Average 
and computation of discrete quantities such as Boolean functions, as we
have illustrated with our results.

Another direction that we find compelling is modelling meta-information
in a more sophisticated way. As of now, we have the single-blind model and
the double-blind model. But there are various intermediate notions that
are reasonable to study. For example, each player might know the number of
players and the number of inputs but nothing about which inputs are given
to which other players. Or a player might know which other players also
receive the inputs it receives, but nothing about inputs it does not
receive. Or some global property of the allotment structure, such as that
the allotment structure is even, might be known. Notice that in the case
of Averaging macroscopes, there \emph{is} an efficient protocol whose cost
doesn't depend on the number of players if the allotment structure is even
and the size of each allotment is known. The protocol simply involves the 
players summing all their inputs and dividing by a universal constant. In
general, one can ask: assuming that there are efficient single-blind protocols
known, what is the minimal information about the allotment structure required
to give efficient protocols?

Ranking is an important subject that we have not addressed in
this paper.  Some interesting problems can be captured as ranking
macroscopes, and we leave their study for further work.

We have seen that in some cases more information hinders making
a global decision, rather than helping.  More generally, why is
allotment structure important?  We are striving to fully understand
how the allotment structure of information affects our ability to
efficiently answer questions that require global information.

\section{Acknowledgements}

S.R.~would like to acknowledge the support of the UK Engineering
and Physical Sciences Research Council (grant number EP/H012338/1)
and the European Commission (TOMSY Grant Agreement 270436, under
FP7-ICT-2009.2.1 Call 6). A.S.~acknowledges the support of EPSRC via
platform grant EP/F028288/1. R.S.~acknowledges the support of EPSRC
via grant EP/H05068X/1.

\end{document}